\documentclass[11pt]{article}
\usepackage{amsmath}
\usepackage{graphicx}
\usepackage{algorithm}
\usepackage{algorithmicx}
\usepackage[noend]{algpseudocode}
\usepackage[colorinlistoftodos]{todonotes}
\usepackage[colorlinks=true, allcolors=blue]{hyperref}
\usepackage{amsthm}
\usepackage{amssymb}

\def\denseformat{
\setlength{\textheight}{9.2in}
\setlength{\textwidth}{7in}
\setlength{\evensidemargin}{-0.2in}
\setlength{\oddsidemargin}{-0.2in}
\setlength{\headsep}{10pt}
\setlength{\topmargin}{-0.3in}
\setlength{\columnsep}{0.375in}
\setlength{\itemsep}{0pt}
}

\newtheorem{theorem}{Theorem}[section]
\theoremstyle{definition}

\theoremstyle{plain}
\newtheorem{lemma}[theorem]{Lemma}

\denseformat
\begin{document}

\title{Distributed Symmetry-Breaking Algorithms for \\ Congested Cliques\footnote{Open University of Israel. Email: leonidb@openu.ac.il ; viktorkh@gmail.com. This research has been supported by ISF grant 724/15 and Open University of Israel research fund.}}
%
%
\author{Leonid Barenboim \and Victor Khazanov
}
%
%
%

\maketitle              
\begin{abstract}
The {\em Congested Clique} is a distributed-computing model for single-hop networks with restricted bandwidth that has been very intensively studied recently. It models a network by an $n$-vertex graph in which any pair of vertices can communicate one with another by transmitting $O(\log n )$ bits in each round. Various problems have been studied in this setting, but for some of them the best-known results are those for general networks. For other problems, the results for Congested Cliques are better than on general networks, but still incure significant dependency on the number of vertices $n$. Hence the performance of these algorithms may become poor on large cliques, even though their diameter is just $1$. In this paper we devise significantly improved algorithms for various symmetry-breaking problems, such as forests-decompositions, vertex-colorings, and maximal independent set. 

We analyze the running time of our algorithms as a function of the arboricity $a$ of a clique subgraph that is given as input. The arboricity is always smaller than the number of vertices $n$ in the subgraph, and for many families of graphs it is significantly smaller. In particular, trees, planar graphs, graphs with constant genus, and many other graphs have bounded arboricity, but unbounded size. We obtain $O(a)$-forest-decomposition algorithm with $O(\log a)$ time that improves the previously-known $O(\log n)$ time, $O(a^{2 + \epsilon})$-coloring in $O(\log^* n)$ time that improves upon an $O(\log n)$-time algorithm, $O(a)$-coloring in $O(a^{\epsilon})$-time that improves upon several previous algorithms, and a maximal independent set algorithm with $O(\sqrt a)$ time that improves at least quadratically upon the state-of-the-art for small and moderate values of $a$. 

Those results are achieved using several techniques. First, we produce a forest decomposition with a helpful structure called {\em $H$-partition} within $O(\log a)$ rounds. In general graphs this structure requires $\Theta(\log n)$ time, but in Congested Cliques we are able to compute it faster. We employ this structure 
in conjunction with
partitioning techniques that allow us to solve various symmetry-breaking problems efficiently.
\end{abstract}

\section{Introduction}
\subsection{The Congested Clique Model and Problems}
In the message-passing {\em LOCAL} model of distributed computing a network is represented by an $n$-vertex graph $G = (V,E)$. Each vertex has its own processing unit and memory of unrestricted size. In addition, each vertex has a unique identity number (ID) of size $O(\log n)$. Computation proceeds in synchronous rounds. In each round vertices perform local computations and send messages to their neighbors. The running time in this model is the number of rounds required to complete a task. Local computation is not counted towards running time. Message size is not restricted. Therefore, this model is less suitable for networks that are constrained in message size as a result of limited channel bandwidth. To handle such networks, a more realistic model has been studied. This is the {\em CONGEST} model that is similar to the LOCAL model, except that each edge is only allowed to transmit $O(\log n)$ bits per round. An important type of CONGEST networks that has been intensively studied recently is the {\em Congested Clique} model. It represents single-hop networks with limited bandwidth. Although the diameter of such networks is $1$, which would make any problem on such graphs trivial in the LOCAL model, in the Congested Cliques various tasks become very challenging. Note that the Congested Clique is equivalent to a general $n$-vertex graph in which any pair of vertices (not necessarily neighbors) can exchange messages of size $O(\log{}n)$ in each round. Such a general graph corresponds to a subgraph of an $n$-clique. The subgraph constitutes the input, while the clique constitutes the communication infrastructure.

The study of the problem of Minimum Spanning Tree (henceforth, MST) was initiated in the Congested Clique model by Lotker et al. \cite{LPPSP2003}. They devised a deterministic $O(\log \log n)$-rounds algorithm that improved a straight-forward $O(\log n)$ solution. In the sequel, randomized $O(\log \log \log n)$-rounds- \cite{HPPSS15},\cite{PS2014}, $O(\log^* n)$-rounds\footnote[1]{$\log^* n$ is the number of times the $\log_2$ function has to be applied iteratively until we arrive at a number smaller than $2$. That is, $\log^* 2 = 1$, and for $n > 2,$ $\log^* n = 1 + \log^* (\log n)$.}-
\cite{GP16}, and $O(1)$-rounds \cite{J17} algorithms for MST in Congested Cliques were devised. These algorithms, however, may fail with certain probabilities. Thus obtaining deterministic algorithms that never fail seems to be a more challenging task in this setting. Since the publication of the result of \cite{LPPSP2003} many additional problems have been studied in the Congested Clique setting \cite{CHKKLPS2015,CHPS2013,Gall17,G17,HP15}. In particular, several {\em symmetry-breaking} problems were investigated. 
Solving such problems is very useful in networks in order to allocate resources, schedule tasks, perform load-balancing, and so on.
Hegeman and Pemmaraju \cite{HP15} obtained a randomized $O(\Delta)$-coloring algorithm with $O(1)$ rounds if the maximum degree $\Delta$ is at least $\Omega(\log^4 n)$, and $O(\log \log n)$-time otherwise. We note that although in a clique it holds that $\Delta = n - 1$, and an $O(\Delta)$-coloring algorithm is trivial (by choosing unique vertex identifiers as colors), the problem is defined in a more general way. Specifically, we are given a clique $Q = (V,E)$, and a subgraph $G' = (V,E'), E' \subseteq E$. The goal is computing a solution for $G'$ as a function of $\Delta = \Delta(G')$, rather then $\Delta(Q)$. In this case the $O(\Delta)$-coloring problem becomes non-trivial at all. We are not aware of previously-known deterministic algorithms for coloring in the Congested Clique that outperform algorithms for general graphs. (Except an algorithm of \cite{CHPS2013} that is not applicable in general, but rather if $\Delta = O(n^{1/3})$. In this case its running time is $O(\log \Delta)$.)

Another symmetry-breaking problem that was studied in the Congested Clique is Maximal Independent Set (henceforth, MIS). The goal of this problem is to compute a subset of non-adjacent vertices that cannot be extended. Again, this problem is interesting in subgrahs of the Congested Clique, rather than the Congested Clique as a whole. A deterministic algorithm for this problem with running time $O(\log \Delta \log n)$ was devised in \cite{CHPS2013}. If $\Delta = O(n^{1/3})$ then the running time of the algorithm of \cite{CHPS2013} improves to $O(\log \Delta)$. Ghaffari \cite{G17} devised a randomized MIS algorithm for the Congested Clque that requires $\tilde{O}(\log \Delta / \sqrt{\log n} + 1) \leq \tilde{O}(\sqrt {\log \Delta})$ rounds. Interestingly, when $\Delta$ is not restricted, all above-mentioned deterministic algorithms and most randomized ones have significant dependency on the clique size $n$. Obtaining a deterministic algorithm for these problems that does not depend on $n$ is an important objective, since very large clique subgraphs may have some bounded parameters (e.g., bounded arboricity) that can be utilized in order to improve running time.

\subsection{Our Results and Techniques}
In this paper we devise improved {\em deterministic} symmetry-breaking algorithms for the Congested Clique that have very loose dependency on $n$, or not at all. Specifically, for clique subgraphs with {\em arboricity}\footnote[1]{The arboricity is the minimum number of forests that graph edges can be partitioned into. It always holds that $a(G') \leq \Delta(G')$, and often the arboricity of a graph is significantly smaller than its maximum degree.} $a$ we obtain $O(a)$-coloring in $O(a^{\epsilon})$ time (for an arbitrarilly small constant $\epsilon > 0$), $O(a^{1 + \epsilon})$-coloring in $O(\log^2 a)$ time, $O(a^{(2 + \epsilon)})$-coloring in $O(\log^* n)$ time and Maximal Independent Set in $O(\sqrt a)$ time. The best previously-known algorithms for these coloring problems are those for general graphs, and incur a multiplicative factor of $\log n$. See table below. Moreover, the $\log n$ factor is unavoidable when solving these problems in general graphs \cite{BE2008}. Our results demonstrate that in Congested Cliques much better solutions are possible. Our MIS algorithm outperforms the results of \cite{CHKKLPS2015} when there is a large gap between $a$ and $\Delta$ or between $a$ and $n$.
For example, trees, planar graphs, graphs of constant genus, and graphs that exclude any fixed minor, all have arboricity $a = O(1)$. On the other hand, their maximum degree $\Delta$ and size $n$ are unbounded.

\begin{table}[h]
\begin{center}
\begin{tabular}{|c|c|c|c|}

\hline 
\multicolumn{2}{|c|}{\textbf{Our Results (Deterministic)}} 
& 
\multicolumn{2}{|c|}{\textbf{Previous Results (Deterministic and Randomized)}}
\\
\hline 
    & Running Time &   & Running Time\\
  \hline
  Forest-Decomposition & $O(\log{}a)$ 
  & Forest-Decomposition \cite{BE2008} & $O(\log{}n)$\\
\hline
$O( a^{2+\varepsilon})$-coloring  & $O(\log^* n)   $ 
  & $O( a^{2+\varepsilon})$-coloring \cite{BE2008} & $O(\log{}n)$\\
\hline
$O( a^2)-$coloring  & $O(\log{}a) +\log^* n $ 
  & $O( a^2)-$coloring \cite{BE2008} & $O(\log{}n)$\\
\hline
$O( a^{1+\varepsilon})$-coloring   & $O(\log{^2}a)   $ 
  & $O( a^{1+\varepsilon})$-coloring \cite{BE2011} & $ O(\log{}a  \log{}n)$\\
\hline
$O( a)$-coloring  &   $O(a^{\varepsilon})$
  & $O( a)$-coloring \cite{BE2011} & $O(\mbox{min}(a^\varepsilon \log{}n,a^\varepsilon+ \log^{{1+\varepsilon}}n))$ \\
\hline
MIS & $O(\sqrt a)$ & MIS \cite{BE2008} & $ O(a + \log n)$ \\
\hline
& & MIS \cite{CHPS2013} & $O(\log \Delta \log n)$ \\
\hline
& & MIS (rand.) \cite{G17} & $\tilde{O}(\sqrt { \log \Delta})$\\
\hline
$O(\Delta)$-coloring & $O(a^{\epsilon})$ & $O(\Delta)$-coloring (rand.) \cite{HP15} & $ O(\log \log n)$ \\
\hline
\end{tabular}
\end{center}
\end{table}

Our main technical tool is an $O(a)$-forests-decomposition algorithm that requires $O(\log a)$ rounds in the Congested Clqiue. This is in contrast to general graphs where $O(a)$-forests-decomposition requires $\Theta(\log n)$ rounds. Once we compute such a forests decomposition, each vertex knows its $O(a)$ parents in the $O(a)$ forests of the decomposition. We orient edges towards parents. The union of all edges that point towards parents constitute the edge set $E'$ of the input. This is because for each edge, one of its endpoint is oriented outwards, and is considered in the union. Note also that the out degree of each vertex is $O(a)$. Then, within $O(a)$ rounds each vertex can broadcast the information about all its outgoing edges to all other vertices in the graph. Indeed, each outgoing edge can be represented by $O(\log n)$ bits using IDs of endpoints. Then, in round $i \in O(a)$, each vertex broadcasts to all vertices the information of its $i$th outgoing edges. After $O(a)$ rounds all vertices know all edge of $E'$ and are able to construct locally (in their internal memory) the input graph $G' = (V,E')$.

Once vertices know the input graph they can solve any computable problem (for unweighted graphs or graphs with weights consisting of $O(\log n)$ bits) locally. The vertices run the same deterministic algorithm locally, and obtain a consistent solution (the same in all vertices). Then each vertex deduces its part from the solution of the entire graph. This does not require communication whatsoever, and so the additional (distributed) running time for this computation is $0$. Thus our results demonstrate that any computable problem can be solved in the Congested Clique in $O(a)$ rounds deterministically. This is an alternative way of showing what follows from Lenzen's \cite{L2013} routing scheme, since a graph with arboricity a has $O(n \cdot a)$ edges that can be announced within $O(a)$ rounds of Lenzen's algorithm. But the additional structure of forests-decomposition that we obtain  is useful for speeding up certain computations, as we discuss below. We note that although in this model it is allowed to make unrestricted local computation, in this paper we do not abuse this ability, and devise algorithms whose local computations are reasonable (i.e., polynomial).

Since any computable problem can be solved in $O(a)$ rounds, our next goal is obtaining algorithms with a better running time. We do so by partitioning the input into subgraphs of smaller arboricity. We note that vertex disjoint subgraphs are Congested Cliques by themselves that can be processed in parallel. For example, partitioning the input graph into $O(a^{1-\epsilon})$-subgraphs of arboricity $O(a^{\epsilon})$, and coloring subgraphs in parallel using disjoint palettes, makes it possible to color the entire input graph with $O(a)$ colors in $O(a^{\epsilon})$ time rather than $O(a)$. Partitioning also works for MIS, although this problem is more difficult to parallelize. (In the general CONGEST model the best algorithm in terms of $a$ has running time $O(a + \log^* n)$.) Nevertheless, using our new partitioning techniques we obtain an MIS with $O(\sqrt a)$ time in the Congested Clique. We believe that this technique is of independent interest, and may be applicable more broadly. Specifically, by quickly partitioning the input into subgraphs of small arboricity, we can solve any computable problem in these subgraphs in $O(a^{\epsilon})$ time, rather than $O(a)$. Given a method that efficiently combines these solutions, it would be possible to obtain a solution for the entire input significantly faster than $O(a)$.

\subsection{Related Work}
Lenzen \cite{L2013} devised a communication scheme for the Congested Clique. Specifically, if each vertex is required to send $O(n)$ meassages of $O(\log n)$ bits each, and if each vertex needs to receive at most $O(n)$ messages, then this communication can be performed within $O(1)$ rounds in the Congested Clique. Algebraic methods for the Congested Clique were studied in \cite{CHKKLPS2015,Gall17}. Symmetry-breaking problems were very intensively studied in general graphs. Many of these results apply to the Congested Clique. In particular, Goldberg, Plotkin, and Shannon \cite{GPS88} devised a	$(\Delta + 1)$-coloring algorithm with running time $O(\Delta \log{}n)$. Goldberg and Plotkin \cite{GP87} devised an $O(\Delta^2)$-coloring algorithm with running time $O(\log{^*}n)$ for constant values of $\Delta$. Linial \cite{L92} extended this result to general values of $\Delta$. Kuhn and Wattenhofer \cite{KW2006} obtained a $(\Delta +1)$ coloring algorithm with running time $O(\Delta \log \Delta + \log{^*}n)$. Barenboim and Elkin \cite{BE2011} devised an $O(\mbox{min}(a^\varepsilon \log{}n,a^\varepsilon+ \log^{{1+\varepsilon}}n))$-time algorithm for $O(a)$-coloring, and 
$ O(\log{}a  \log{}n)$-time algorithm for $O(a^{1+\varepsilon})$-coloring.

\section{Preliminaries}
We provide some definitions and survey several known procedures that are needed for our algorithms that we describe in the next sections. 
We relegate descriptions of known procedures to Appendix A. This includes $H$-partitions, Forests-Decomposition, Defective-coloring, $O(a)$-proper-coloring and Lenzen's routing schem in the Congested Clique. Readers that are familiar with these concepts may proceed directly to Section \ref{sc:fdec} after reading Section \ref{sc:dfntion}.
\subsection{Definitions} \label{sc:dfntion}
The {\em $k$-vertex-coloring} problem is defined as follows.
Given a graph $G=(V,E)$,  find a proper coloring  $\varphi:V\to $ {1,2,...,k} 
that satisfies $\varphi(v) \neq  \varphi(u),   \forall (u,v) \in E$.
 The \textit{out-degree} of a vertex $v$ in a directed graph is the number of edges incident to $v$ that are oriented out of $v$. 
 An \textit{orientation} $\mu$ of (the edge set of) a graph is an assignment of direction to each edge $(u,v) \in E$ either towards $u$ or towards $v$. Consider a graph $G = (V,E)$ in which some of the edges are oriented. 
In our work we use a concept of {\em partial orientations}, which was employed  by Barenboim and Elkin \cite{BE2011}. A partial orientation  is allowed not to orient some edges of the graph. By this definition, 
a partial orientation $\sigma$ has
deficit at most $d$, for some positive integer parameter $d$, if
for every vertex $v$ in the graph the number of edges incident
to $v$ that $\sigma$ does not orient is no greater than $d$.  Another
important parameter of a partial orientation is its length $l$. This is the length of the longest path $P$ in which all edges are oriented consistently by $\sigma$. (That is, each vertex in the path has out-degree and in-degree at most $1$ in the path.)
An {\em $H$-partition} $(H_1,H_2,...,H_{\ell}) $  of $G=(V,E)$ with degree $A$, for some parameter $A$, is a partition of $V$, such that for any vertex in a set $H_i$, $i \in [\ell]$, the number of its neighbors in $H_i \cup H_{i + 1} \cup ... \cup H_{\ell}$ is at most $A$.

\def\AppA{
\subsection{H-partition}
\label{arbhpartsec}
The arboricity $a = a(G)$ is the minimum number a of edge-disjoint forests $F_1,F_2,...,F_a$ whose union covers the entire edge set $E$ of the graph $G = (V,E)$. Such a decomposition is called an $a$-forest-decomposition of G.
The structure of $H$-partitions is useful for computing forests decompositions.
A procedure for computing an $H$-partition, called {\em Procedure Partition}, was devised in \cite{BE2008}. This procedure accepts as input the arboricity  of the graph, and an arbitrarily small positive real constant  $\varepsilon \leq$ 2. The parameter $\varepsilon$ determines the quality of the resulting $H$-partition. This means that smaller values of $\varepsilon$ result in a better partition, but require more time.
\textit{Procedure Partition} computes an $H$-partition with degree at most
$(2+\varepsilon) \cdot a$ and size $l=\lceil \frac{2}{\varepsilon} \log{}n \rceil$ within $l$ rounds.
During the execution of Procedure Partition each vertex in \(V\) is either active or inactive. Initially, all the vertices are active. For every $i=1,2,...,l$ in the $i$th round each
active vertex with at most $(2+\varepsilon)\cdot a$ active neighbors joins the set $H_i$  and becomes inactive. The following results were proven in \cite{BE2008}.

\begin{lemma}
\label{lem1}
\cite{BE2008}  For a graph $G$ with arboricity $a(G) = a$, and a parameter $\varepsilon$,
$0<\varepsilon\leq 2$, Procedure Partition $(a, \varepsilon)$ computes an H-partition of size $l=\lceil \frac{2}{\varepsilon} \log{}n \rceil$ with degree at most $(2+\varepsilon)\cdot a$. The running time of the procedure is $O(\log{}n)$.

\end{lemma}

\begin{lemma}
\label{H-partition}
\cite{BE2008}  For a graph $G$ with arboricity $a(G) = a$, and a parameter $\varepsilon$,
$0<\varepsilon\leq 2$, $G$ has at least $\frac{\varepsilon}{\varepsilon +2} \cdot |V|$  vertices  with degree  $(2+\varepsilon)\cdot a$ or less .

\end{lemma}

\begin{lemma}
\label{H-partition}
\cite{BE2008}    For any subgraph G' of G, the arboricity of G' is at most the arboricity of G.

\end{lemma}

\begin{lemma}
{\cite{BE2008}  } 
\label{H-partition}
 The H-partition $H = {H_1,H_2,...,H_l }$,   $l\leq  \lceil \frac{2}{\varepsilon  } \log{}n \rceil  $    , has degree at most $A=(2+\varepsilon) \cdot a$
\end{lemma}

\subsection{Forests-Decomposition}
Coloring oriented forests can be performed extremely efficiently in the distributed setting, both in terms of running time and in the number of colors.
For a wide range of graph families, it is possible to color oriented graphs significantly faster than a coloring of general graphs, using the decomposition to forests. If a graph can be decomposed
into a reasonably small number of oriented forests, then both the running time and the size of the employed coloring palette can be reduced. 
A {\em $k$-forests-decomposition} is a partition of the edge set of the graph into $k$ subsets, such that each subset forms a forest. Efficient distributed algorithms for computing $O(a)$-forests decompositions have been devised recently in \cite{BE2008}  . Several results from \cite{BE2008}   are used in this work.

\begin{lemma} 
{\cite{BE2008}  }
(1) For any graph G, a proper  $ (\lfloor ( 2+\varepsilon)\cdot a \rfloor +1 ) $  -coloring of G can be computed in $O(a \log{}n)$  time, for an
arbitrarily small positive constant $\varepsilon$.
\\
(2) For any graph G, an $O(a)$ - forest-decomposition can
be computed in $O(\log{}n)$ time.
\end{lemma}

Another, useful, procedure is Procedure Arb-Linial \cite{BE2008}  . Which is essentially a composition of Linial \cite{L92}  $O(\Delta^2)$-coloring algorithm  with an algorithm Procedure Forests-Decomposition \cite{BE2008}  . The main difference of the coloring step of Procedure Arb-Linial from the original Linial coloring algorithm is that in Procedure Arb-Linial each vertex considers only the colors of its parents in forests
$F_1,F_2,...,F_A$, where 
$A\leq \lfloor (2+\varepsilon) \cdot a \rfloor $ rather than all its neighbors.

\begin{lemma}
\label{lemma-arblinial}
{\cite{BE2008}  } An $ O(a^2)-coloring$ can be computed in $O(\log{^*n})$ time
\end{lemma}

\subsection{Defective coloring }

An \textit{m-defective p-coloring} of a graph \textit{G} is a coloring of the
vertices of \textit{G} using \textit{p} colors, such that each vertex has at
most \textit{m} neighbors colored by its color. Each color class in
the \textit{m-defective coloring} induces a graph of maximum degree
\textit{m}.

It is known that for any positive integer parameter $p$, an 
$ \lfloor \frac{\triangle}{p}  \rfloor$-defective $O(p^2)$ -coloring can be efficiently computed distributively \cite{BEK2014}.

\begin{lemma} 
\cite{BEK2014}
$ \lfloor \frac{\triangle}{p}  \rfloor $-defective $O(p^2)$ -coloring can be computed in $O(\log{^*}n)$ time 
\end{lemma}

An \textit{r-arbdefective k-coloring} is a coloring
with \textit{k} colors, such that all the vertices colored by the same
color \textit{i}, $1 \leq i \leq k$, induce a subgraph of \textit{arboricity} at most \textit{r}.
Barenboim and Elkin \cite{BE2011} devised an efficient procedure for computing an arbdefective coloring \textit{Arbdefective-Coloring Procedure}. The procedure, receives as input a graph \textit{G} and two positive integer parameters \textit{k} and \textit{t}. 

Barenboim and Elkin \cite{BE2011} defined a procedure \textit{Simple-Arbdefective} which works in the following way. The procedure accepts as
input such an orientation and a positive integer parameter $t$. During its execution, each vertex computes its color as follows. Each vertex waits for its parents to
select their colors. Once the vertex receives a message from each of 
its parents containing their selections, it selects a color from
the range $1,2,...,k$ that is used by the minimum number of
parents. Then it sends its selection to all its neighbors. This completes the description of the procedure.
It is used in a more sopisticated procedure called Arbdefective-Coloring. Its properties are summarized below.


\begin{lemma} 
{\cite{BE2011}}
Procedure Arbdefective-Coloring invoked on
a graph $G$ with arboricity $a$, and two positive integer parameters $k$ and $t$, computes an  $ (\lfloor a/t + ( 2+\varepsilon)\cdot a/k  \rfloor ) $   -arbdefective $k$-coloring in time  $O(t^2 \log{}n)$. 
\end{lemma}

\subsection{O(a)-coloring }
The O(a)-coloring algorithm of Barenboim and Elkin \cite{BE2011} works as follows.
The procedure receives as input a graph G and a positive integer parameter p. It proceeds in phases. In the first phase Procedure Arbdefective-Coloring  is invoked on the input graph $G$ with the parameters k=p and t=p. Consequently, a decomposition into p subgraphs is produced, in
which each subgraph has arboricity $O(a/p)$. In each of the
following phases Procedure Arbdefective-Coloring is invoked
in parallel on all subgraphs of the  decomposition of the previous phase. Each subgraph is partitioned into $p$ sub-graphs of smaller arboricity. Thus, after each phase, the number of subgraphs in $G$ grows by a
factor of $p$, however the arboricity of each subgraph shrinks by a factor of $\Theta(p)$. Consequently, the product of the number of sub-graphs and the arboricity of subgraphs remains $O(a)$ after each phase. Once the arboricities of all subgraphs become small enough, this is used for a fast parallel coloring of all the sub-graphs, resulting in a  proper $O(a)$-coloring of the  graph $G$.

\begin{lemma} 
{\cite{BE2011}}
Invoking Procedure Legal-Coloring on a graph $G$ with arboricity $a$ with the parameter $p= \lceil  a^{\frac{\mu}{2}} \rceil $  for a
positive constant $\mu < 1$, produces a legal $O(a)$-coloring of $G$
within $O(a^{\mu} \cdot \log{}n)$  time. 
\end{lemma}

\subsection{Lenzen's routing algorithm}
One of the important building blocks for algorithms in the Congested Clique
model is Lenzen's routing algorithm  \cite{L2013}. This algorithm guarantees that if there is a component of an algorithm in which each node needs to send at most
$O(n \log{}n)$ bits and receive at most $O(n \log{}n)$ bit, then $O(1)$ rounds are sufficient. This corresponds to sending and receiving $O(n)$ pieces of data with a size $O(\log{}n)$ to every node. Intuitively, this is easy when each piece of information of a node has a distinct destination, via a direct message. Since, source-destination partition does not have to be uniform , it is a big advantage of Lenzen's algorithm.

\begin{lemma} 
\label{lemlenzen}
{\cite{L2013}}
The Algorithm  of Optimal Deterministic Routing provides a routing scheme such that if each node is the source for $O(n)$ messages and each node
is the  designation for $O(n)$ messages, then all of these
messages can be routed from their sources to their destinations within $O(1)$ rounds.

\end{lemma}

}

\section{Forest-Decomposition-CC} \label{sc:fdec}
In this section we describe our Forest-Decomposition algorithm for the Congested Clique.
Our Forest-Decomposition algorithm starts with computing an $H$-partition.  This computation is performed faster in Congested Cliques than in general graphs thanks to the following observation. Once the first $O(\log{a})$ $H$-sets are computed (within $O(\log{a})$ time), the subgraph induced by the remaining active vertices has at most $O(n)$ edges. (We prove this in Lemma  ~\ref{fastfdc}  below.) Consequently, all these vertices can learn this entire subgraph using Lenzen's algorithms within O(1) rounds. Then each vertex can locally compute the $H$-set it belongs to. This is in contrast to the algorithm for general graphs where the running time is $\Theta(\log{n})$, even for graphs with $O(n)$ edges. 

First we provide a procedure which computes an $H$-partition within $O(1)$ rounds, on graphs with edge set of size at most $O(n)$. This procedure is based on Lenzen's routing scheme. The main idea of the procedure is that each vertex can transmit all edges adjacent on it to all other vertices in the graph. This is because the overall number of messages each vertex receives in this case is $O(n)$. Indeed, each edge can be encoded as a message of size $O(\log n)$ that contains the IDs of the edge endpoints, and the number of messages is bounded by the number of edges in the graph. Since the number of sent messages of each vertex is also bounded by $O(n)$, Lenzen's scheme allows all vertices to transmit all their edges to all other vertices within constant number of rounds, as long as the number of edges is $O(n)$. Once a vertex receives all the edges of the graph, it constructs the graph in its local memory. All vertices construct the same graph, and perform a local computation of the $H$-partition. This does not require any communication whatsoever, but since all vertices hold the same graph, the resulting $H$-partition is consistent in all vertices. This completes the description of the procedure. 


\begin{algorithm}
\caption{$H$-partition   of an input graph $G$ with arboricity $a$ and  $O(n)$  edges }
\label{alg: Procedure Sparse-Partition }
\begin{algorithmic}[1]
\Procedure{   Sparse-Partition}{$G,a,\varepsilon$} 

\State  Each node $u$ in $G$ broadcasts its degree to every other node $v$ in $G$
\State  Using Lenzen's scheme, send all information about all edges to all vertices of $G$
\State Each vertex $v \in V$ perfomrs locally the following operations:
\State Initially, all vertices of $G$ are marked as active.
\State $i= \left \lceil \frac {2}{\varepsilon} \log a  + 1 \right \rceil$ 
\While{$i \leq \frac{2}{\varepsilon}\log{}n $}
\If{ $v$ is active and has at most $(2 + \varepsilon )\cdot a$ active neighbors}
\State 	        make $v$ inactive
\State add $v$ to $H_i$
\EndIf
\State $i = i + 1$
\EndWhile\label{H-Partitionendwhile}
\EndProcedure
\end{algorithmic}
\end{algorithm}

\pagebreak

Next, we provide a general procedure to compute an $H$-partition in graphs with any number of edges in the Congested Clique model. The preocedure is called {\em Procedure $H$-Partition-CC}. The computation is done by first reducing the number of edges to $O(n)$ within $O(\log{a})$ rounds, and then invoking Procedure Sparse-Partition on the remaining subgraph. The reduction phase (lines 3 - 13 of the algorithm below) operates similarly to Procedure Sparse-Partition, but the partition into $H$-sets is performed in a distributed manner, rather than locally, and the number of iterations is just $O(\log a)$, rather than $O(\log n)$. In the next lemmas we show that this is sufficient to reduce the number of edges to $O(n)$.  

\begin{algorithm}
\caption{Computing an $H$-partitions of a general graph $G$ with arboricity $a$  in the Congested Clique model}
\label{alg:Procedure H-Partition-CC}
\begin{algorithmic}[1]
\Procedure{ H-Partition-CC}{$a,\varepsilon$} 
\State \textit{An algorithm for each vertex v ∈ V :}
\State $i=1$
\While{$i \leq \left \lceil \frac{\varepsilon}{2} \cdot \log{}a \right \rceil  $}
\If{ $v$ is active and has at most $(2 + \varepsilon )\cdot a$ active neighbors}
\State 	        make $v$ inactive
\State add $v$ to $H_i$
\State 	 send the messages "inactive" and "$v$ joined $H_i$" to all the neighbors
\EndIf
\State 	 	\textbf{for} each received "inactive" message \textbf{do}
\State     \hspace{1cm}   mark the sender neighbor as inactive
\State \textbf{end for}
\State $i = i + 1$

\EndWhile 
\State \textbf {end while}

\State $H_i$,$H_{i + 1}$...,$H_{O(\log{}n)}$ = invoke Procedure Sparse-Partition on the subgraph induced by remaining active vertices  
\EndProcedure
\end{algorithmic}
\end{algorithm}

\begin{lemma} \label{fastfdc}
After  $\left \lceil \frac{2}{\varepsilon}\log{}a \right \rceil$ rounds (lines 4-13 in Algorithm 2), the number of edges whose both endpoints are incident to nodes that are still active is $O(n)$.
\end{lemma}

\begin{proof}
Consider the $i$th  iteration. By Lemma ~\ref{H-partition} in Appendix A, the graph $G_i$ induced by the remaining active vertices in the 
round \textit{i} has  ${(\frac{2}{2+\varepsilon})}^i  \cdot |V| $  vertices. 
Recall that a graph with arboricity a has no more than $n \cdot a$ edges. The number of edges in the graph $G_i$  is at most: 
${(\frac{2}{2+\varepsilon})}^i  \cdot n \cdot a$ .  Then in the round  $i= \left \lceil \frac{2} {\varepsilon}\log{}a \right \rceil$, the graph $G_i$ has  ${(\frac{2}{2+\varepsilon})}^{\left \lceil \frac{2} {\varepsilon}\log{}a \right \rceil}  \cdot n \cdot a $ = $O(n)$ edges.
\end{proof}
The next lemma states the correctness of Algorithm 2, as well as its running time.
\begin{lemma}
Algorithm 2 computes an $H$-partion in $O(\log{a})$ rounds.
\end{lemma}
\begin{proof}
The correctness of Algorithm 1 follows from the correctness of H-partition of  \cite{BE2008} in conjunction with Lenzen's routing scheme. Specifically, within $O(\log{a})$ rounds the algorithm properly computes the $H$-sets $H_1$,$H_2$,...,$H_{O(\log a)}$, and within an additional round the remaining subgraph is learnt by all vertices using Lenzen's scheme, and all $H$-sets of this subgraph, up to 
$H_{O(\log{n})}$, 
are computed locally by each vertex. Thus, each vertex can deduce the index of its $H$-set within $O(\log{a})$ rounds from the beginning of the algorithm.

\end{proof}

We summarize the properties of Procedure H-Partition-CC in the following theorem:
\begin{theorem}
\label{theorem1}
Procedure H-Partition-CC  invoked on a graph \textit{G} with arboricity\textit{ a(G)} and a parameter $\varepsilon$, $0 < \varepsilon \leq $2 computes an $H$-partition of size $l=O(\log{}n)$ with degree at most $O(a)$.
The running time of the procedure is $O(\log{}a)$.
\end{theorem}

We next devise a forest-decomposition algorithm for the Congested Clique model, called \textit{Procedure Forest-Decomposition-CC}. It accepts as input the parameters $a$ and $\varepsilon$. In the first step, it computes an \textit{H-Partition-CC}, with degree at most $(2+ \varepsilon) \cdot a$ .
In the next step, it invokes a procedure called Procedure Orientation \cite{BE2011} as follows. 	
\\Procedure Orientation:
 For each edge $e=(u,v)$, if the endpoints $u,v$ are in different sets
$H_i,H_j,i \neq j $, then the edge is oriented towards the vertex in the set with a greater index. 
Otherwise, if $i = j$, the edge $e$ is oriented towards the vertex with a greater \textit{ID} among the two 
vertices \textit{u} and \textit{v}. The orientation $\mu$ produced by this step is acyclic. Each vertex 
has  out-degree at most $(2+ \varepsilon) \cdot a$.
The correctness of the procedure follows from the correctness of
Procedure Orientation from \cite{BE2008}.\\
The last step of the algorithm is partitioning the edge set of the graph into forests as follows: 
each vertex is in charge of its outgoing edges, and it assigns each outgoing edge a distinct label from the 
set $\{1,2,...,(2+ \varepsilon) \cdot a\}$. This completes the description of the algorithm. Its pseudocode and analysis are provided below.

\begin{algorithm}
\caption{ Partitioning of the edge set of $G$ into  $ (\lfloor ( 2+\varepsilon)\cdot a \rfloor ) $ forests in the Congested-Clique model}
\label{alg:Procedure Forests-Decomposition-CC}
\begin{algorithmic}[1]
\Procedure{Forests-Decomposition-CC }{$a,\varepsilon$} 
\State invoke Procedure $H$-Partition-CC($a$, $\varepsilon$) 
\State  $\mu$ = Orientation() 
 
\State  assign a distinct label to each $\mu$-outgoing edge of \textit{v} from the set $ [\lfloor ( 2+\varepsilon)\cdot a \rfloor ] $
\EndProcedure
\end{algorithmic}
\end{algorithm}

\begin{lemma}
\label{lem14}
The time complexity of Procedure Forests-Decomposition-CC is
$O(\log{}a)$.
\end{lemma}
\begin{proof} 
Procedure H-Partition-CC  takes $O(\log{}a)$ time, and steps (2) and (3) of Forests-Decomposition-CC    require O(1) rounds each. Therefore, the overall time of Procedure Forests-Decomposition-CC  is $O(\log{}a)$.
\end{proof}

\begin{theorem}
\label{theorem2}
For a graph G with arboricity $a = a(G)$, and a parameter  
$\varepsilon ,0<\varepsilon \leq 2$, 
in Congested Clique, Procedure Forests-Decomposition-CC $(a,\varepsilon)$ partitions the edge set 
of \textit{G} into $ (\lfloor ( 2+\varepsilon)\cdot a \rfloor ) $ forests in $O(\log{}a)$ rounds. Moreover, as a result of its execution each 
vertex \textit{v} knows the label and the orientation of every edge $(v,u)$ adjacent to \textit{v}.
\end{theorem}

\section{A general solution with $O(a)$ time in Congested Clique}
\label{sc:gsln}
In this section we describe how to solve any computable problem in $O(a)$ time in the Congested Clique. We note that since any graph with arboricity $a$ has $O(a\cdot n)$ edges, this is possible to achieve by directly applying $O(a)$ rounds of Lenzen's scheme \cite{L2013}. However, in this section we present an alternative solution that employs forest-decompositions.
Given a forest-decomposition in which the number of parents (i.e. outgoing edges) of each vertex is bounded by $O(a)$, we can solve any computable problem within this number of rounds. 
Specifically, once Procedure Forests-Decomposition-CC is invoked, it partitions the edge set of $G$  into  $ (\lfloor ( 2+\varepsilon)\cdot a \rfloor ) $   forests in $O(\log{}a)$ rounds. As a result of its execution, each vertex $v$ knows the label and the orientation of every edge $(v,u)$ adjacent to $u$. An outgoing edge from a vertex $v$ to a vertex $u$ labeled with a label $i$ means that $u$ is the parent of  $v$ in a tree of the $i$th forest $F_i$. Therefore,
by transmitting the information of a distinct parent in a round, each vertex can inform all other vertices of the graph about all its parents. This will require an overall of $O(a)$ rounds - one round per parent. Then, each vertex knows all parents of all vertices in the graph $G$. But this information is sufficient to constuct the graph $G$ locally. Indeed, for each edge $e$ of the graph $G$, one of its enpoints is a parent of the other in some forest $i$, and thus this edge is announced to all vertices in round $i$. Within $O(a)$ rounds, all edges are announced, and so the entire graph is known to all vertices.
Therefore,  we can solve any computable problem on $G$ locally (without any additional communication), by executing the same deterministic algorithm on the same graph that is known to all. This guarantees a consistent solution in all vertices. Thus, we obtain a general solution with $O(a)$ time to any computable problem in the Congested Clique. (Note that this is true either if the input graph $G$ is unweighted or if $G$ has weights on edges that require $O(\log n)$ bits per edge. In the latter case, the information about weights can be transmitted together with the information about parents withouth affecting the running time bound $O(a)$. Recall, however, that all our algorithms in this paper are for unweighted graphs.)
Therefore, it would be more interesting to find faster than $\Theta(a)$ algorithms for various problems. We obtain such algorithms in the next sections.

\section{$O(a^2)$-coloring in $O(\log a + \log^* n)$ time}

Note that in the synchronous message-passing model of
distributed computing  a proper $O(a^2)$-
coloring requires $\Theta(\log{}n)$ time \cite{L92}.
However,  in Congested Clique we can improve the running time  and reach even better result  of $O(\log{}a)+ \log{^*}n$. 
 
In this section we employ  Procedure Forests-Decomposition-CC to provide an efficient algorithm that colors the input graph G  of arboricity \textit{a=a(G)} in $O(a^2)$ colors. The running time of the algorithm is 
$O(\log{}a)+\log{^*}n$.
For computing an $O(a^2 )$-coloring we will use Procedure Arb-Linial described in  \cite{BE2008}. 
Procedure Arb-Linial accepts  a graph $G$ with arboricity $a(G)$. Given an $O(a)$-forests-decomposition of $G$, the procedure computes a proper coloring $\varphi$ of the graph using $O(a^2)$ colors in $O(\log{^*n})$ running time. During the execution of this procedure, each vertex transmits at most $O(\log n)$ bits over each edge in each round.

Procedure Forest-Decomposition-CC has better running time than the respective procedure on general graphs, which allows us to compute a proper $O(a^2)$-coloring of the graph very quickly. We devise a procedure  called \textit{Procedure Arb-Coloring-CC}   that works in the following way. The procedure starts by executing Procedure Forest-Decomposition-CC with the 
input parameter $a=a(G)$. This invocation returns an \textit{H}-partition of $G$ of size $l\leq   \lceil \frac{2}{\varepsilon  } \log{}n \rceil  $, and degree at most $A=(2+\varepsilon) \cdot a$.
Then, we invoke Procedure Arb-Linial on the forest-decomposition. Since the procedure requires each vertex to send only its current color to its neighbors (which is of size $O(\log n)$), Procedure Arb-Linial can be invoked as-is in the congested clique. In our case we execute Procedure Arb-Linial with an input  parameter
$A=(2+\varepsilon) \cdot a$.
In Procedure Arb-Linial each vertex considers only the colors of its parents in forests $F_1, F_2, ..., F_A$. By Lemma ~\ref{lemma-arblinial} in Appendix A the algorithm computes $O(((2+\varepsilon)\cdot a)^2 ) = O(a^2)$-coloring.
This completes the  description of Procedure Arb-Coloring-CC. Its pseudocode and running time analysis are provided below.

\begin{algorithm}
\caption{$O(a^2)$-coloring in the Congested Clique}
\label{alg:Procedure  Arb-Coloring-CC}
\begin{algorithmic}[1]
\Procedure{  Arb-Coloring-CC}
{$a,\varepsilon$}

\State 	$H=(H_1,H_2,…,H_l)$ =  invoke Procedure Forest-Decomposition-CC

\State invoke Procedure Arb-Linial ($H$, $A=(2+\varepsilon) \cdot a$)
\EndProcedure
\end{algorithmic}
\end{algorithm}

\begin{theorem}
\label{theorem3}
 Procedure Arb-Coloring-CC computes a proper $O(a^2)$-coloring in the Congested Clique in $O(\log{}a + \log{^*}n)$ rounds.
\end{theorem}
\begin{proof} 
The correctness of the procedure follows from the above discussion. 
The running time of step (1) is  $O(\log{}a)$ rounds, by Lemma ~\ref{lem14}. Step (2),  by Lemma ~\ref{lemlenzen}, requires $O(\log{^*}n)$ rounds. 
Thus, the overall running time of the procedure is  $O(\log{}a)+ \log{^*}n$.
\end{proof}  

\section{$O(a^{2+\varepsilon}) $-coloring in $O(\log^* n)$ time}
In this section we show that the factor of $\log a$ can be eliminated from the running time of Theorem ~\ref{theorem3} in the expense of slightly increasing the number of colors to $O(a^{2+\varepsilon})$, for an arbitrarilly small positive constant $\varepsilon$. To this end, we invoke Procedure H-Partition-CC
with second parameter set as 
${a^\varepsilon}$, rather than $\varepsilon$. We show below that this way the running time of forests-decompositions becomes just $O(1)$. However, the number of forests produced is now $O(a^{(1+\varepsilon)})$, rather than $O(a)$. Moreover, once Procedure Forest-Decomposition-CC terminates, we invoke  Arb-Linial-CC algorithm on the result of the forest decomposition to compute $O((a^{(1+\varepsilon)})^2)$-Coloring.  

\begin{lemma}
\label{lemma-H-Partition-CC}
Invoking Procedure H-Partition-CC with the second parameter set as $q=a^\varepsilon$ requires $O(1)$ rounds.
\end{lemma}
\begin{proof} 
In each round the number of   active vertices is reduced by a factor of $\Theta(a^{\varepsilon})$.
For $i = 1,2,...$, the number of edges in the subgraph  induced by active vertices in round $i$ is  at most $O(\frac{(a \cdot n)}{(a^\varepsilon)^i})	$.     
Thus, after $i=O(\frac{1}{\varepsilon} )$ rounds,the number of remaining edges will be $O(n)$. Then we can employ Lenzen's scheme, broadcast these edges to all vertices within $O(1)$ rounds, and compute the remaining $H$-sets locally.
Therefore, the overall running time is $O(\frac{1}{\varepsilon}) =  O(1)$.
\end{proof}

\begin{lemma}
\label{forestdc}
For graphs \textit{G} with \textit{ a(G)=a}, and a parameter, $q=a^\varepsilon$, for an arbitrarilly small positive constant $\varepsilon$, Procedure Forest-Decomposition-CC partitions the edge set of \textit{G} into
$A=O(a^{1+\varepsilon})$ oriented forests in $O(1)$ rounds in Congested Clique.
\end{lemma}
\begin{proof} 
 
 By Lemma ~\ref{lemma-H-Partition-CC}, Procedure \textit{H-Partitions-CC} executes in $O(1)$ rounds, the second stage is an orientation that is computed in \textit{O(1)} rounds, and assigning labels to outgoing  edges is computed in \textit{O(1)} rounds as well. Therefore, the overall time of is $O(1)$. 
\end{proof}

The next theorem follows directly from Lemmas \ref{lemma-H-Partition-CC} - \ref{forestdc}.
\begin{theorem}
\label{theorem4}
For graphs $G$ with $a(G)=a$ and with a parameter $q=a^\varepsilon$, for a positive constant $\varepsilon$, Procedure Arb-Coloring-CC computes $O(a^{(2+\varepsilon)}  )$-coloring within  $O(\log^* n)$ time in Congested Clique.
\end{theorem}

\section{$O(a^{1+\varepsilon})$-coloring in $O(\log^2 a)$ time}

In this section we devise an algorithm that produces $O(a^{1+\varepsilon})$-coloring in $O(\log{^2}a + \log{^*}n)   $  running time. We employ a combination of defective colorings  and forest decompositions. Usually, when a vertex is required to select a color, it chooses a color different from the colors of all its neighbors. The vertex's neighbors select their colors in different rounds. Alternatively,  in a defective coloring, a vertex can select a color that is already used by  its neighbors. Furthermore,  neighbors can perform the selection in the same round. Therefore, the computation can  be significantly more efficient. Moreover, defective colorings allow us to obtain helpfull structures with appropriate properties, such as partial orientations with small deficit.
We start by presenting a procedure, called Procedure Partial-Orientation-CC. It is based on a procedure from \cite{BE2008}, but the current variant is adapted to Congested Cliques, and it is more efficient than the variant for general graphs. The procedure receives as an input a graph $G$ and an  integer  $t>0$. It computes an orientation with out-degree 
$\lfloor (2+\varepsilon) \cdot a \rfloor$ 
and a deficit at most $\lfloor  \frac{a}{t}  \rfloor$. (Recall that the deficit is the maximum number of unoriented edges adjacent on the same vertex.) 

Procedure Partial-Orientation-CC contains   three steps.
First, an $H$-partition of the input graph $G$ is computed. Then the vertex set of $G$ is partitioned into subsets $H_1,H_2,...,H_l$, such that every vertex in $H_i, 1 \leq i \leq O(\log{}n)$, has $O(a)$ neighbors in 
$ \bigcup_{j=i}^{\log{}n}H_j$. In the next step,   
 $ (\lfloor a/t    \rfloor ) $ -defective $O(t^2)$-coloring is
computed in each $G(H_i)$ in parallel, using \cite{BE2011}. The final step is  a computation of an orientation  as follows. Consider an edge $e=(u,v), u \in H_i, v \in H_j  $ for some $1 \leq i,j \leq O(\log{}n)$. If $i<j$, orient the edge towards $v$. If $j<i$, orient the edge towards $u$. Otherwise $i=j$. In this case the vertices $u$ and $v$ may have
different colors or the same color. If the colors are different, orient the edge towards the vertex that is
colored with a greater color.
Otherwise, the edge remains unoriented.
This complete the describing of the procedure.

 \begin{algorithm}
\caption{Computing a partial orientation with length $O(t^2 \log n)$ and deficit $a/t$ in the Congested Clique}
\label{alg:Procedure Partial-Orientation-CC}
\begin{algorithmic}[1]
\Procedure{ Partial-Orientation-CC}
{$G,t$}
\State  	$H=(H_1,H_2,...,H_l)$ Invoke Procedure H-Partition-CC

\State  \textbf{For} each $i=1,...,\log{}n$  in parallel \textbf{do}: 

\State \hspace{1cm} compute an $ (\lfloor a/t    \rfloor ) $ -defective $O(t^2)$  -coloring of $G(H_i)$ 

\State  \textbf{For} each edge $e = (u,v)$ in $E$ in parallel \textbf{do}: 
 \hspace{1cm} \If{  $u$ and $v$ belong to different $H$-sets}
\State 	       orient e towards the set with greater index.

\ElsIf { $u$ and $v$ have different colors} 
\State orient $e$ towards the vertex with greater color between u, v.

\EndIf
 
\EndProcedure
\end{algorithmic}
\end{algorithm}

 \begin{lemma}
For graphs \textit{G} with \textit{ a(G)=a}, with parameters $ \varepsilon , 0<\varepsilon \leq 2$  and integer \textit{t}, $t>0$. The Procedure Partial-Orientation-CC produces an acyclic orientation of out-degree   $\lfloor (2+\varepsilon) \cdot a \rfloor$
\end{lemma}
\begin{proof} 
 
 Consider a vertex $v \in H_i$. Each outgoing edge of $v$ is connected to a vertex in a set $H_j$ such that $j \geq i$. By Lemma 2, $v$ has at most
  $\lfloor (2+\varepsilon) \cdot a \rfloor$ neighbors in $ \bigcup_{j=i}^{\log{}n}H_j$. Thus, the out-degree of $v$  is at most $\lfloor (2+\varepsilon) \cdot a \rfloor$.
\end{proof}

\begin{lemma}
For graph \textit{G} with \textit{ a(G)=a}, with parameters $ \varepsilon , 0<\varepsilon \leq 2$  and integer \textit{t}, $t>0$. The Procedure Partial-Orientation-CC produces an acyclic orientation of length   $O(t^2 \cdot \log{}n)$
\end{lemma}
\begin{proof} 
 Consider a directed path $p'$ in $G(H_i)$. The length of $p'$ is smaller than the number of colors used in the defective coloring of $G(H_i)$, which is  $O(t^2)$. (This is because each edge on a path is directed towards a greater color, and the number of colors of $H_i$ is $O(t^2)$.) Consider a directed path $p$ in $G$  with respect to the orientation produced by Procedure Partial-Orientation-CC. The path $p$ contains at most $O(\log{}n)$ edges which cross between different $H$-sets. (This is because each edge that cross between $H$-sets is directed towards a greater index, and the number of indices of $H$ sets is bounded by $O(\log n)$.) Note, that  between any pair of such edges that cross between $H$-sets, there are at most $O(t^2)$ edges which belong to the same $H$-set (with respect to both their endpoints). Therefore, the length of the path $p$ is at most $O(t^2 \cdot \log{}n)$.
\end{proof}

\begin{theorem}
\label{theorem4}
The running time of  the Procedure Partial-Orientation-CC 
 on a graphs \textit{G} with \textit{ a(G)=a}, with parameters $ \varepsilon , 0<\varepsilon \leq 2$  and integer \textit{t}, $t>0$ is $O(\log{}a + \log{^*n})$
\end{theorem}
\begin{proof} 
 The first step of Procedure Partial-Orientation-CC is 
 Procedure H-Partition-CC which requires $O(\log{}a)$ rounds. The second step, is computing defective colorings, which by Lemma 6, requires $O(\log{^*}n)$ time. Orientation step requies only $O(1)$ rounds.  Thus, the overall time is $O(\log{}a + \log{^*n})$.
\end{proof}

Partial Orientations allow us to compute arbdefective colorings as follows. Each vertex waits for all neighbors on outgoing edges (henceforth, parents) to select a color from a certain range $\{1,2,...,k\}$. Then a vertex selects a color that is used by the minimum number of parents. While this is not a proper coloring, it partitions the graph into subgraphs induced by color classes. These subgraphs have smaller arboricity, and can be processed more efficiently. By repeating this several times, we obtain subgraphs with sufficiently small arboricity that can be colored directly. Then we combine all colorings efficiently to obtain a unified coloring of the input graph. This general scheme was developed in \cite{BE2011} for general graphs. But here we apply it more efficiently on Congested Cliques, using their special properties and the new techniques we devised for them. 

Once we defined Procedure Partial-Orientation-CC, we proceed to Procedure Simple-Arbdefective \cite{BE2011} to compute $O(a/k)$-arbdefective $k$-koloring. In other words, it computes a vertex decomposition into $k$ subgraphs such that each subgraph has
arboricity $O(a/k)$.  (See Appendix A). Note that in the first round the vertices without outgoing edges have nothing to wait for, and so they are colored in the first round. 


\begin{algorithm}
\caption{Computing $O(a/k)$-arbdefective $k$-coloring}
\label{alg:Procedure Simple-Arbdefective-CC}
\begin{algorithmic}[1]
\Procedure{ Simple-Arbdefective-CC }
{$G,k$} 

\State An algorithm for each vertex $v \in V$ 
\State  \textbf	{While} ($v 
$ is not colored  ) \textbf{do} 
 \State \hspace{0.5cm}  \textbf{if} each parent $u$ of $v$ is colored \textbf{then} 
\State 	\hspace{0.8cm}$v$ selects a color from the range {$1,2,...,k$},  used by the minimum number of parents.
\State 	\hspace{0.8cm}  send the messages "$v$ is colored"  to all the neighbors
\State  \textbf	{end While}
\EndProcedure
\end{algorithmic}
\end{algorithm}

	Now, we define our next procedure, called Procedure Arbdefective-Coloring-CC. The procedure receives as input a graph $G$ and two positive integer parameters $k$ and $t$. First, it invokes Procedure Partial-Orientation-CC on $G$ and $t$. After that it employs the produced orientation and the parameter $k$ as an input for Procedure Simple-Arbdefective-CC, which is activated as soon as Procedure Partial-Orientation-CC ends.
Note that during the invocation of Procedure Partial-Orientation-CC an execution of Lenzen's scheme is performed, and so all vertices learn the subsets $\{H_j, H_{j + 1},...,H_l\}$, $j = \Theta(\log a), l = O(\log n)$ of the $H$ partition. 
We will refer to partition $H_j,H_{j+1}...,H_l$ as a subpartition $H'$ of $H = \{H_1,H_2,...,H_l\}$.	
Once the partition $H'$ becomes known to all vertices, Procedure Simple-Arbdefective-CC can be invoked on it locally, without communication whatsoever. Then any vertex that belongs to $H_i$, $i \geq j$, selects its color immediately, according to this computation. Vertices in $H_i$ with $i < j$ must select their colors by executing a distributed algorithm. This is done again using Procedure Simple-Arbdefective-CC, but since the number of remaining $H$-sets is just $O(\log a)$, this is done more efficiently than invokig it on the entire graph. This completes the description of the procedure.
Its pseudocode is provided below.
 \begin{algorithm}
\caption{Computing an arbdefective coloring with $k$ colors and arbdefect $O(a/t + a/k)$ in the Congested Clique}
\label{alg:Procedure Partial-Orientation-CC}
\begin{algorithmic}[1]
\Procedure{ Arbdefective-coloring-CC}
{$G,k,t$}
\State  	$H=\{H_1,H_2,...,H_l\}$ invoke Procedure Partial-Orientation-CC($G$, $t$)

\State let $H' = \{H_j, H_{j + 1},...,H_l\}$, $j = \Theta(\log a)$, be the sets that all vertices $v \in V$ have learnt as a result of the invocation of line 2 

\State invoke Procedure Simple-Arbdefective-CC($G$, $k$) locally on $H'$

\State invoke Procedure Simple-Arbdefective-CC($G$, $k$) in a distributed manner on $H \setminus H' = \{H_1,H_2,...,H_{j-1}\}$
 
\EndProcedure
\end{algorithmic}
\end{algorithm}


 \begin{lemma}
The running time of Procedure Arbdefective-Coloring-CC
is $O(t^2\cdot\log{}a)$
\end{lemma}
\begin{proof} 
 Procedure Partial-Orientation-CC requires $O(\log{}a)$ rounds. Then all vertices learn the sets  $H' = \{H_j, H_{j + 1},...,H_l\}$ and color them locally within $O(1)$ rounds. Consequently, any remaining oriented path of uncolored vertices belongs to $H \setminus H'$, and thus has length $O(t^2 \log a)$. Indeed, a path may consists of at most $O(\log a)$ edges that cross between $H$-sets of $H \setminus H'$, and at most $O(t^2)$ edges that are within the same $H$-set between pairs of crossing edges.
\end{proof}

\begin{lemma} 
Procedure Arbdefective-Coloring-CC invoked on 
a graph $G$ and two positive integer parameters $k$ and $t$ computes an   
$ (\lfloor a/t + ( 2+\varepsilon)\cdot a/k  \rfloor ) $   -arbdefective 
 $k$ -coloring in time  $O(t^2 \log{}a)$ 
\end{lemma}

\begin{proof}
The number of outgoing edges of each vertex is at most $(2 + \varepsilon) \cdot a$. (See Lemma ~\ref{lem1}.) Consider a subgraph $G_i$ induced by vertices of the same color $i \in \{1,2,...,k\}$. Since each vertex selected a color used by minimum number of parents from the set $\{1,2,...,k\}$, it has at most $(2 + \varepsilon) a / k$ outgoing edges in $G_i$. (By pigeonhole principle.) In addition, a vertex in $G_i$ may have at most $a/t$ unoriented edges adjacent on it in $G_i$, since the deficit is at most $a/t$. For the purpose of analysis we can add directions to all unoriented edges, such that the graph remains acyclic. This is done by a topological sortng of vertices according to directions of originally oriented edges. Then each vertex in $G_i$ has out degree at most $ (\lfloor a/t + ( 2+\varepsilon)\cdot a/k  \rfloor ) $, all edges are oriented, and the orientation is acyclic. Hence the arboricity of $G_i$ is at most $ (\lfloor a/t + ( 2+\varepsilon)\cdot a/k  \rfloor ) $ for all $i \in \{1,2,...,k\}$.
\end{proof}

We will invoke Procedure Arbdefective-Coloring-CC  with a parameter $t = k = O(1)$, that has to be a sufficiently large constant. In this case it returns a 
$ ( ( 3+\varepsilon)\cdot a/t  \rfloor ) $  -arbdefective t-coloring in
 $O(t^2 \log{}a)$ time.
Such a $t$-coloring constitutes a decomposition of $G$ into $t$ sub-graphs with arboricity at most $ ((3+\varepsilon)\cdot a/t   ) $ in each of them.
The invocations are performed by a procedure we define next. The procedure is called Procedure Proper-Coloring-CC. The main idea is partitioning an input graph $G$ into subgraphs $G=G_1,G_2,...,G_k$ using  Procedure Arbdefective-Coloring-CC in time  $O(t^2 \log{}a)$, and then invoking Procedure Proper-Coloring-CC recursively on these subgraphs. Note that each vertex-induced subgraph of a Congested Clique is a Congested Clique by iteslf, and so it is possible to invoke Procedure Proper-Coloring-CC recursively. The number of recursion levels is going to be $O(\log{}a)$, and thus the overall running time is $O(\log^2 a)$. Our ultimate goal is to partition an input graph $G$ by Procedure Arbdefective-Coloring-CC to subgraphs 
 $G_i, 1\leq i \leq a^{1 + \varepsilon}$ with arboricity $a(G_i)=O(1)$. This is the termination condition of the recursion. 
 In the bottom level of the recursion, when all subgraphs have a constant arboricity, we invoke our general algorithm from Section 4 to color the subgraphs with $O(1)$ colors each, in constant time.
 We apply this idea in the following Procedure Proper-Coloring-CC.

  \begin{algorithm}
\caption{Proper Coloring in Congested Clique}
\label{alg:Recursion-Arbdefective-Coloring-CC}
\begin{algorithmic}[1]
\Procedure{Proper-Coloring-CC}
{$G\;' , \alpha$}

 \State $p = $ a sufficiently large constant

 \If{  $\alpha > p$}
 
 \State  \textbf{for} each $G_i \in G\;' $  in parallel \textbf{do}: 
\State  \hspace{0.3cm}  $ G\;''_1,G\;''_2,...G\;''_{p} $ = Procedure Arbdefective-Coloring-CC($G_i,k=p,t=p $) 
\State \hspace{0.3cm}   $\alpha $=$ (3 + \varepsilon) \frac{\alpha}{p}   $ /* New upper bound for arboricity of each subgraph */
\State  \hspace{0.3cm} Proper-Coloring-CC$(G_i , \alpha)$
\State  \textbf{end for}
\Else
\State Color each $G_i \in G\;'$ using our general algorithm of Section \ref{sc:gsln} with $O(\alpha)$  distinct colors  \ \  \ \ \ \ \  \  \ \   /* \ \ \ \ \ \ $O(\alpha) = O(p) = O(1)$ \ \ \ \ \ */
\EndIf

\EndProcedure
\end{algorithmic}
\end{algorithm}
 
 \pagebreak

 The procedure receives as input a graph $G$. In each recursion level Procedure Arbdefective-Coloring-CC is invoked on an input graph $G\;'$. Then a decomposition into $p$  subgraphs is performed, where each subgraph has arboricity at most $(3 + \varepsilon)( \frac{a(G\;')}{p} )$. In each of the following recursion levels,  Procedure Arbdefective-Coloring-CC is called in parallel on all subsequent subgraphs that were created at the previous levels. As a result, a refinement of the decomposition is obtained, that is, each subgraph partitioned further into $p$ subgraphs of yet lower arboricity.  Consequently, after each level, the number of subgraphs in $G$ grows by a factor $p$, but the arboricity  of each subgraph decreases by a factor of $p/(3 + \varepsilon)$. Consequently, in level $i$ of the recursion, the product of the number of subgraphs and the arboricity of subgraphs is $O((3 + \varepsilon)^i \cdot a)$. Once the arboricity of each graph becomes at most $p$, the procedure  terminates in a level denoted $r$, and returns an $O((3 + \varepsilon)^r \cdot a)$-coloring of the entire graph. (Since there are $O((3 + \varepsilon)^r \cdot a)$ subgraphs in that stage, and each is colored with distinct $O(1)$-colors.) We next analyze the procedure.

 \begin{lemma}
In the end of level i of the recursion , $i = 1, 2,..$  any
graph $G''_j$ that is produced in this level has arboricity at most 
$((3 + \varepsilon)/p)^i \cdot a(G)$, where $a(G)$ is the arboricity of the original input graph $G$.

\end{lemma}
\begin{proof} 
The proof is by induction on the number of levels. The base case is the first level. Then $G$ is partitioned into $p$ subgraphs produced by Procedure Procedure Arbdefective-Coloring-CC, with arboricy at most $(3 + \varepsilon) a / p$ in each of them.   For the inductive step, consider a level $i$. By the induction hypothesis, each subgraph in $G\;'$ has arboricity at most $((3 + \varepsilon)/p)^{(i - 1)} \cdot a(G)$ . During level $i$, Procedure Arbdefective-Coloring-CC is invoked on all subgraphs in  $G\;'$. Consequently, the new subgraphs have arboricity at most $(3 + \varepsilon)((3 + \varepsilon)/p)^{(i - 1)} \cdot a(G)/p = ((3 + \varepsilon)/p)^i \cdot a(G)$.
\end{proof}

 \begin{lemma}
The recursion proceeds for  $(\log a) / (\log (p/(3 + \varepsilon))) $  levels.
\end{lemma}
\begin{proof} 
In each level the parameter $\alpha$ is decreased by a multiplicative factor of $p/(3 + \varepsilon)$, for a sufficiently large constant $p$. Therefore, the number of levels is at most 
 $\log{}_{p/(3 + \varepsilon)}a =(\log a) / (\log (p/(3 + \varepsilon))) $. 
\end{proof}

 Once the arboricity of each graph become $O(1)$, we  color each subgraph properly using $O(1)$ distinct colors. Note that it is indeed possible to use distinct palettes for each subgraph so each vertex deduces an appropriate color (i.e., different from colors of other subgraphs and from neighbors in the same subgraph) using the index of the vertex's subgraph, and the indexes of subgraph collections the vertex belongs to in the recursion tree. The following theorem analyses the running time of the procedure.

 






\begin{theorem}
\label{theorem6}
The running time of  the Procedure Proper-Coloring-CC  
 on a graphs $G$ with arboricity $a(G)=a$   is $O(\log{^2}a)$. The procedure colors an input graph $G$ with $O(a^{1+\varepsilon})$ colors, for an arbitrarilly small positive constant $ \varepsilon$.
\end{theorem}
\begin{proof} 
There are $\log a$ recursion levels, each level requires $O(t^2 \log a) = O(\log a)$ rounds. The bottom levell requires $O(1)$ time. Thus, overall, the running time is $O(\log^2 a)$. The number of colors is  $((3 + \varepsilon))^{\log a / (\log (p/(3 + \varepsilon)))} \cdot a(G) = a^{1 + \varepsilon}$, for a sufficiently large constant $p$.
\end{proof}

\section{$O(a)$-coloring in $O(a^{\varepsilon})$ time}

Our goal in this section is to efficiently compute an $O(a)$-coloring of the graph $G$. In  Proper-Coloring-CC we invoked Procedure Arbdefective-Coloring-CC on  a graph $G$ with the input parameters $p = k = t = O(1)$.  If we invoke our Proper-Coloring-CC algorithm with different parameters,
$p = k = t =a^\varepsilon$, for an arbitrarily small constant $\varepsilon > 0$, we obtain the following result.

 \begin{theorem}
 \label{theorem7}
 
Invoking  Procedure Proper-Coloring-CC on a graph $G$ with arboricity $a$ with the parameter $p= \lceil  a^{\varepsilon/3} \rceil $, produces a proper $O(a)$-coloring of $G$
within   $O(a^{\varepsilon})$   time.
\end{theorem}
\begin{proof} 
 During the execution of Procedure Proper-Coloring-CC, the number of recursion levels is $O(3/\varepsilon)$, i.e., a constant. In each level, the number of colors increases just by  a constant as well. In each level  Procedure Partial-Orientation-CC is executed, which requires $O(t^2 \log{}a) = O(a^{2 \varepsilon / 3} \log a) = O(a^{\varepsilon})$ rounds. The bottom recursion level requires $O(a^{\varepsilon/3})$ time and produces $O(a^{\varepsilon/3})$-coloring in each of the $O(a^{1- \varepsilon})$ subgraphs of this stage, using  our algorithm from Section 4. Hence, the total running time is $O(a^{\varepsilon})$.  
\end{proof}

\section{MIS}
Our MIS algorithm works in the following way. We invoke procedure Proper-Coloring-CC with $p = a^{1/8}$ instead of a constant. Moreover, we perform recursive calls as long as $\alpha > p^4$, rather than $\alpha > p$. Consequently, there are just four recursion levels, each of which requires $O(t^2 \log a) = O(a^{1/4} \log a)$ time. At the bottom level of the recursion, each subgraph has arboricity $O(\sqrt a)$ and there are $q = O(\sqrt a)$ such subgraphs. Denote these subgraphs $G_1,G_2,...,G_q$. Once Procedure-Coloring-CC terminates, we perform the following loop consisting of $q$ iterations. For each $i = 1,2,...,q$, we compute an MIS locally in $G_i$. At this stage all vertices of the subgraph $G_i$ have already learnt it during the execution of Procedure Proper-Coloring-CC, so this is indeed possible. Once vertices join MIS, they send their neighbors a message telling them not to join. Each vertex that receives a message from a neighbor in the MIS, broadcasts to all vertices in the graph that it is outside the MIS. Once these messages are received, each vertex deletes such vertices that have neighbors in the MIS from each $G_1,G_2,...,G_q$ in its local memory. This completes the description of an iteration. Once all iterations complete, we have an MIS of the entire graph. The pseudocode of the algorithm is provided below.
Next, we analyze its correctness and running time.
  \begin{algorithm}
\caption{MIS in Congested Clique}
\label{alg:Recursion-Arbdefective-Coloring-CC}
\begin{algorithmic}[1]
\Procedure{MIS-CC}
{$G' , \alpha$}
\State initially, M = $\emptyset$

 \State compute a decomposition into $q = O(\sqrt \alpha)$ subgraphs $G_1,G_2,...,G_q$ of arboricity $O(\sqrt \alpha)$
 
 \State Each vertex in each $G_i$, $i = 1,2,...,q$, learns the subgraph $G_i$ using our general algorithm from Secion \ref{sc:gsln}.
 
\For {i = 1,2,...,q}

\State compute an MIS of $G_i$ locally and add its vertices to $M$

\State each vertex of $G_i$ that is in $M$ broadcasts this information to all vertices

\State each vertex that has a neighbor in $M$ broadcasts this infomation to all vertices

\State each vertex removes in its local memory the vertices of $G_{i + 1},G_{i + 2},...,G_q$ that have neighbors in $M$
\EndFor
\State {\bf end for}

\State return $M$

\EndProcedure
\end{algorithmic}
\end{algorithm}

\begin{theorem}
Procedure MIS-CC computes a proper MIS of the input graph.
\end{theorem}
\begin{proof}
We prove that for $i = 1,2,...,q$, after iteration $i$, the subgraph of $G'$ induced by vertices of $G_1,G_2,...,G_i$ has a proper MIS. The proof is by induction on $i$.\\
{\bf Base ($i = 1$):} After the first iteration an MIS of $G_1$ is computed and added to $M$.\\
{\bf Step:} In the beginning of iteration $i$, by induction hypothesis, $M$ contains an MIS of the subgraph induced by vertices of $G_1,G_2,...,G_{i - 1}$. In iteration $i - 1$, all neighbors of $M$ announced this to all other vertices, and as a result were removed from $G_i$ in the local memories of processors. Consequently, the MIS that is computed in iteration $i$ on line 6 of the procedure does not have neighbors in the set $M = M_{i - 1}$ produced in the end of iteration $i - 1$. Denote the MIS computed in iteration $i$, step 6, by $M'$. It follows that $M_{i - 1} \cup M'$ is an independent set. Moreover, any vertex in $G_1,G_2,...,G_{i - 1}$ is at distance at most $1$ from some vertex in $M_{i-1}$, by induction hypothesis. Any vertex in $G_i$ is at distance $1$ from some vertex in $M_{i-1}$ (if it was removed from $G_i$), or at distance at most $1$ from some vertex in $M'$ (if it remained in $G_i$). Thus the set $M$ computed after $i$ iterations, $M_i = M_{i - 1} \cup M'$, is an MIS of vertices of $G_1,G_2,...,G_i$.
\end{proof}

\begin{theorem}
The running time of Procedure MIS-CC is $O(\sqrt a)$.
\end{theorem}
\begin{proof}
The decomposition in line 3 is obtained using Procedure Proper-Coloring-CC that is invoked with $p = a^{1/8}$, instead for a constant. It is invoked for 4 recursion levels. Consequently, in the bottom level, the arboricity of each subgraph is $O(\sqrt a)$. Hence, the running time of step 3 is $O(t^2 \log a + \sqrt a) = O(\sqrt a)$. In line 4, each subgraph of $G_1,G_2,...,G_q$ is learnt by all its vertices. This is performed in parallel for $i = 1,2,...,q$, and requires $O(\sqrt a)$ time, since the arboricity of each subgraph is $O(\sqrt a)$. Each iteration of the loop in lines 5 - 10 requires $O(1)$ rounds. Indeed, the computation of MIS is local and does not require communication rounds whatsoever. Broadcasting information of vertices in the MIS requires 1 round. (Each vertex broadcasts a message of $O(\log n)$ bits containing its ID.) Broadcasting information about vertices that have neighbors in the MIS also requires 1 round. Therefore, the running time of $q$ iterations is $O(q) = O(\sqrt a)$. This is also the running time of the entire algorithm.
\end{proof}

\bibliographystyle{alpha}

\clearpage
\pagenumbering{roman}
\appendix
\centerline{\LARGE\bf Appendix}
\section{Preliminaries - Basic Procedures}
\AppA

\end{document}